\documentclass[letterpaper]{article}

\newcommand{\longversion}[1]{#1}
\newcommand{\shortversion}[1]{}

\usepackage{amssymb,amsmath,amsthm}
\usepackage{enumerate}
\usepackage{graphicx}
\usepackage{xspace}
\usepackage{tikz}
\usepackage{vmargin}
\setmarginsrb{1in}{0.9in}{1in}{1.1in}{0cm}{0cm}{8mm}{8mm}

\DeclareMathOperator{\operatorClassP}{P}
\newcommand{\classP}{\ensuremath{\operatorClassP}}
\DeclareMathOperator{\operatorClassNP}{NP}
\newcommand{\classNP}{\ensuremath{\operatorClassNP}}
\DeclareMathOperator{\operatorClassFPT}{FPT}
\newcommand{\classFPT}{\ensuremath{\operatorClassFPT}}
\DeclareMathOperator{\operatorClassW}{W}
\newcommand{\classW}[1]{\ensuremath{\operatorClassW[#1]}}

\DeclareMathOperator{\operatorClassXP}{XP}
\newcommand{\classXP}{\ensuremath{\operatorClassXP}}

\usepackage{thm-restate}

\newcommand{\SR}{\textsc{SR}\xspace}
\newcommand{\SRK}[1]{\textsc{SR}$_{ #1 }$\xspace}
\newcommand{\NMTS}{\textsc{NMTS}\xspace}
\newcommand{\SCD}{\textsc{SCD}\xspace}
\newcommand{\WSR}{\textsc{WSR}\xspace}
\newcommand{\catFull}{\textsc{Splits Reconstruction for Caterpillars of Un\-boun\-ded Hair-Length and Maximum Degree 3}}
\newcommand{\catShort}{\textsc{SRC}}
\newcommand{\weights}{\omega: V \rightarrow \mathbb{N}}
\newcommand{\iP}{I_P}
\newcommand{\iC}{I_C}

\newcommand{\WSRK}[1]{\textsc{WSR}$_{ #1 }$\xspace}

\newcommand{\set}[1]{\{ #1 \}}
\newcommand{\mywlog}{without loss of generality\xspace}
\newcommand{\myWlog}{Without loss of generality\xspace}

\newtheorem{theorem}{Theorem}

\newtheorem{definition}[theorem]{Definition}

\newtheorem{lemma}[theorem]{Lemma}
\newtheorem{prop}{Property}
\newtheorem{claimLNCS}{Claim}
\title{Complexity of Splits Reconstruction\\ for Low-Degree Trees%
\thanks{A preliminary version of this article appeared in the proceedings of WG 2011 \cite{GaspersLSS11}.
The authors acknowledge the support of Conicyt Chile via projects Fondecyt 11090390 (M.L., K.S.),
Fondecyt 11090141 (M.S.), Anillo ACT88 (K.S.), and Basal-CMM (S.G., M.S., K.S.).
The first author acknowledges partial support from the European Research Council (COMPLEX REASON, 239962).
The second and fourth authors acknowledge the support of
the French Agence Nationale de la Recherche (ANR AGAPE ANR-09-BLAN-0159-03)} %optional
}

\author{Serge Gaspers\thanks{%
Institute of Information Systems, Vienna University of Technology, Vienna, Austria.
E-mail: \texttt{gaspers@kr.tuwien.ac.at}}
\and Mathieu Liedloff\thanks{%
LIFO, Universit{\'e} d'Orl{\'e}ans, Orl{\'e}ans, France.
E-mail: \texttt{mathieu.liedloff@univ-orleans.fr}}
\and Maya Stein\thanks{%
CMM, Universidad de Chile, Santiago, Chile.
E-mail: \texttt{mstein@dim.uchile.cl}}
\and Karol Suchan\thanks{%
FIC, Universidad Adolfo Ib{\'a}{\~n}ez, Santiago, Chile.
E-mail: \texttt{karol.suchan@uai.cl}} ~\thanks{%
WMS, AGH - University of Science and Technology, Krakow, Poland.}}

\date{}

\begin{document}

\maketitle

\begin{abstract}
Given a vertex-weighted tree $T$, the split of an edge $x y$ in $T$ is $\min
\set{s_x(x y), s_y(x y)}$ where $s_u(u v)$ is the sum of all weights of
vertices that are closer to $u$ than to $v$ in $T$.
Given a set of weighted vertices $V$ and a multiset of splits
$\mathcal{S}$, we consider the problem of constructing a tree on $V$ whose
splits correspond to $\mathcal{S}$. The problem is known to be
\classNP-complete, even when all vertices have unit weight and the maximum
vertex degree of $T$ is required to be no more than $4$. We show that
\begin{itemize}
 \item the problem is strongly \classNP-complete when $T$ is
required to be a path,
 \item the problem is \classNP-complete when all vertices have unit
weight and the maximum degree of $T$ is required to be no more than $3$, and
 \item it remains \classNP-complete when all vertices have unit weight and $T$ is
required to be a caterpillar with unbounded hair length and maximum degree at
most $3$.
\end{itemize}
We also design polynomial time algorithms for
\begin{itemize}
\item the variant where $T$ is required to be a path and the number of distinct vertex weights is constant, and
\item the variant where all vertices have unit weight and $T$ has a constant number of leaves.
\end{itemize}
The latter algorithm is
not only polynomial when the number of leaves, $k$, is a constant, but also fixed-parameter tractable
when parameterized by $k$.

Finally, we shortly discuss the problem when the vertex
weights are not given but can be freely chosen by an algorithm.

\medskip
\noindent
The considered problem is related to building libraries of chemical compounds
used for drug design and discovery. In these inverse problems, the goal is to
generate chemical compounds having desired structural properties, as there is a
strong correlation between structural properties, such as the Wiener index,
which is closely connected to the considered problem, and biological activity.
\end{abstract}

\section{Introduction}
\label{sec:intro}\label{sec:prelim}

In this paper, we consider trees $T=(V,E)$ where integer weights are associated to vertices by a function $\weights$, where $\mathbb{N}$ denotes the set of natural numbers excluding $0$.

\begin{definition}
Let $T$ be a tree and $\weights$ be a function. The \emph{split} of an edge $e$ in $T$ is the minimum of $\Omega(T_1)$ and $\Omega(T_2)$, where $T_1$ and $T_2$ are the two trees obtained by deleting $e$ from $T$, and $\Omega(T_i)=\sum_{v\in T_i}\omega(v)$. We use $\mathcal{S}(T)$ to denote the multiset of splits of $T$.
\end{definition}

\noindent
We consider the problem of reconstructing a tree with a given multiset of splits and a given set of weighted vertices.

%\medskip
\begin{quote}
\textsc{Weighted Splits Reconstruction} (\WSR):  Given a set $V$ of $n$ vertices, a weight
function $\weights$, and a multiset $\mathcal{S}$ of integers, is there a tree $T$ whose multiset of splits is $\mathcal{S}$ (i.e. $\mathcal{S}(T)=\mathcal{S}$)?
\end{quote}

%\smallskip
%

\noindent
The \textsc{Weighted Splits Reconstruction for Trees of Maximum Degree $k$}  problem (\WSRK{k}) is defined in the same way, except that
we restrict \longversion{the tree }$T$ to have maximum degree at most $k$. When we require $T$ to belong to a class of trees $\mathcal{T}$, the problem is called \textsc{Weighted Splits Reconstruction for $\mathcal{T}$}.\longversion{

}
When $\omega$ assigns unit weights to the vertices, the problem is simply called \textsc{Splits Reconstruction} (\SR). The \textsc{Splits Reconstruction for Trees of Maximum Degree $k$}  problem (\SRK{k}) and the \textsc{Splits Reconstruction for $\mathcal{T}$}  are the obvious unweighted counterparts of the weighted variants defined above.

\smallskip

\noindent \textbf{Related Work.}
In the field of Chemical Graph Theory \cite{Balaban76,Bonchev91,Trinajstic92}, molecules are modeled by graphs in order to study the physical properties of chemical compounds. A chemical graph is a graph, where vertices represent atoms of a chemical compound and edges the chemical bonds between them.  Within the area of quantitative structure-activity relationship (QSAR), several structural measures of chemical graphs were identified that quantitatively correlate with a well defined process, such as biological activity or chemical reactivity. Probably the most widely known example is the \emph{Wiener index} (see \cite{Hammer97}): the sum of the distances in a graph between each pair of vertices, where the distance between two vertices is the \longversion{length (the number of edges) of }\shortversion{number of edges on }a shortest path from one to the other. Wiener \cite{Wiener47} found a strong correlation between the boiling points of paraffins and the Wiener index. From then on, many other topological (using the information of the chemical graph) and topographical (using the information of the chemical graph and the location of its vertices in space) indices were introduced and their correlation with various other biological activities was investigated.

In Combinatorial Chemistry, drug design is facilitated by building libraries of molecules that are structurally related (via the Wiener index or any of the other numerous indices). We face inverse problems where the goal is to design new compounds that have a prescribed structural information (see also \cite{FaulonB10}).

Goldman et al.~\cite{GoldmanILPW00} study problems related to the design of combinatorial libraries for drug design from an algorithmic and complexity-theoretic point of view, following the heuristic approaches of \cite{SheridanK95} and \cite{GilletWBG99}. \shortversion{They }\longversion{Goldman et al.~}show that for every positive integer $W$, except $2$ and $5$, there exists a graph with Wiener index $W$. They also show that every integer, except a finite set, is the Wiener index of some tree. For constructing a tree (of unbounded or bounded maximum degree) with a given Wiener index, they devise pseudo-polynomial dynamic programming algorithms. Goldman et al.~also introduce the \textsc{Splits Reconstruction} problem and recall a result due to Wiener \cite{Wiener47}: the Wiener index of a tree $T$ on $n$ vertices with unit weights is $\sum_{s\in \mathcal{S}(T)} s\cdot (n-s)$. They show that $\SR$ is \classNP-complete and give an exponential-time algorithm without running time analysis.
% (a simple analysis would give a running time of $O(n^n)$).

As it is not reasonable to construct chemical trees with arbitrarily high vertex degrees, Li and Zhang \cite{LiZ04} studied \SRK{4} and showed that it is also \classNP-complete. Their algorithm to construct a tree with maximum degree at most $4$ to solve \SRK{4} runs in exponential time (no running time analysis is provided)
%, but a simple analysis would give a running time of $O(n^n)$)
and creates weighted vertices in intermediate steps.

In order to reconstruct glycans or carbohydrate sugar chains, Aoki-Kinoshita et al. \cite{AokiKinoshitaKKLW06} study the reconstruction of a node-labeled supertree from a set of node-labeled subtrees. They give a 6-approximation algorithm for this problem, which generalizes the smallest superstring problem.\longversion{

}
We refer to \cite{DobryninEG01} surveying results on the Wiener index for trees.

\smallskip

\noindent \textbf{Our Results.}
By the result of Li and Zhang \cite{LiZ04}, \SRK{4} is \classNP-complete, while \SRK{2} is trivially in \classP. We close this gap by showing that \SRK{3} is \classNP-complete by a reduction from \textsc{Numerical Matching with Target Sums} (defined below). It is even \classNP-complete for caterpillars with unbounded hair length. Identifying small classes of trees for which the problem is \classNP-complete may be important for future investigations in the spirit of the deconstruction of hardness proofs \cite{KomusiewiczNU09} which aim at identifying parameters for which the problem becomes tractable when these parameters are small.

Our main result proves that \WSRK{2} is strongly \classNP-complete by a reduction from a variant of \textsc{Numerical Matching with Target Sums} in which all integers of the input are distinct. For the case where the weights of the vertices are chosen from a small set of values, our dynamic-programming algorithm solves \WSRK{2} in time $O(n^{k+3} \cdot k)$, where $k$ is the number of distinct vertex weights.
Although this running time is polynomial for every constant $k$, the degree of the polynomial depends on $k$. Thus, the running time becomes impractical, even for small values of $k$.
Parameterized complexity \cite{DowneyF99,FlumG06,Niedermeier06} is a theoretical framework that allows to distinguish between running times of the form $f(k) n^{g(k)}$ where the degree
of the polynomial depends on the parameter $k$ and running times of the form $f(k) n^{O(1)}$ where the exponential explosion of the running time is restricted to the parameter only.
The fundamental hierarchy of parameterized complexity classes is
\begin{align*}
 \classFPT \subseteq \classW{1} \subseteq \classW{2} \cdots \subseteq \classXP,
\end{align*}
where a parameterized problem is in \classFPT\ (fixed-parameter tractable) if there is a function $f$ such that the problem can be solved in time $f(k) n^{O(1)}$,
a problem is in XP if there are functions $f,g$ such that the problem can be solved in time $f(k) n^{g(k)}$,
and \classW{t}, $t\ge 1$, are parameterized intractibility classes giving strong evidence that a parameterized problem that is hard for any of these classes is not in FPT.
Our algorithm for \WSRK{2} parameterized by the number of distinct vertex weights places this problem in \classXP.
A generalization of this problem is \classW{1}-hard \cite{FellowsGR10}, but it remains open
whether this problem is fixed parameter tractable.
As a relevant parameter for \SR we identified $k$, the number of leaves in the reconstructed tree. This parameterization of \SR can be solved in time $O(8^{k \log k}\cdot n)$, and is thus
fixed-parameter tractable.

\smallskip

\noindent\textbf{Definitions.}
A \emph{caterpillar} is a tree consisting of a path, called its \emph{backbone}, and paths attached with one end to the backbone. Its \emph{hair length} is the maximum distance \longversion{(in terms of the number of edges) }from a leaf to the closest vertex of the backbone. A \emph{star} $K_{1,k}$ is a tree with $k$ leaves and one internal vertex, called the \emph{center}.
In our hardness proofs, we reduce from the following problem (problem [SP17] in \cite{GareyJ79}).

%\smallskip
\begin{quote}
\textsc{Numerical Matching with Target Sums} (\NMTS):  Given three disjoint multisets $A$,$B$, and $S=\set{s_1, \hdots, s_m}$, each containing $m$ elements from $\mathbb{N}$, can $A\cup B$ be partitioned into $m$ disjoint sets $C_1, C_2, \hdots, C_m$, each containing exactly one element from each of $A$ and $B$, such that, for $1\le i\le m$, $\sum_{c\in C_i} c = s_i$?
\end{quote}
%\medskip
%

\shortversion{\noindent Due to space constraints, the proofs of the statements marked with \textbf{($\star$)} have been moved to the appendix.}

\longversion{
\smallskip

\noindent \textbf{Organization.}
The remainder of this paper is organized as follows. Section \ref{sec:hardnessPaths} shows that \WSRK{2} is \classNP-complete. On the positive side, we show in Section \ref{sec:algoPaths} that \WSRK{2} can be solved in polynomial time when the number of distinct vertex weights is bounded by a constant. That this result cannot be extended to \WSRK{3} is shown in Section \ref{sec:hardnessDeg3}: \SRK{3} remains \classNP-complete.
% even for caterpillars with unbounded hair length.
Section \ref{sec:algoLeaves} gives an \classFPT-algorithm for \SR parameterized by the number of leaves of the reconstructed tree.
The variant where the vertex weights are freely choosable is discussed in Section \ref{sec:freeweights} and we conclude with some directions for future research in Section \ref{sec:concl}.
}

\section{\WSRK{2} is strongly \classNP-complete}
\label{sec:hardnessPaths}

In this section, we show that \WSRK{2} is strongly \classNP-complete. First we introduce a new problem that is polynomial-time-reducible to \WSRK{2}, and then show that this new problem is strongly \classNP-hard.

%\smallskip
\begin{quote}
\textsc{Scheduling with Common Deadlines} (\SCD): Given $n$ jobs with positive integer lengths $j_1, \hdots, j_n$ and $n$ deadlines $d_1 \le \hdots \le d_n$,
can the jobs be scheduled on two processors $P_1$ and $P_2$ such that at each deadline there is a processor that finishes a job exactly at this time,
and processors are never idle between the execution of two jobs?
\end{quote}
%\smallskip

\noindent
To reinforce the intuition on this problem one may imagine that we want to satisfy delivery deadlines and avoid using any warehouse space to store a product between its fabrication and the delivery date. There is no restriction as to which product should be delivered at a given time.
(Another possibility is imagining computer scientists scheduling paper production to fit conference deadlines.)

Given an instance $(j_1, \hdots, j_n, d_1, \hdots, d_n)$ for \SCD, we construct an instance for \WSRK{2} as follows. For each job $j_i$, $1\le i\le n$, create a vertex $v_i$ with weight $\omega(v_i) = j_i$. For each deadline $d_i$, $1\le i\le n-1$, create a split $d_i$. We may assume that $\sum_{i=1}^n j_i = d_{n-1}+d_n$, otherwise we trivially face a \textsc{No}-instance.

Suppose the path $P=(v_{\pi(1)}, v_{\pi(2)}, \hdots, v_{\pi(n)})$ is a solution to \WSRK{2}. Say $\{v_{\pi(\ell)},v_{\pi(\ell+1)}\}$ is the edge associated to the split $d_{n-1}$. We construct a solution for \SCD by assigning the jobs $j_{\pi(1)}, j_{\pi(2)}, \hdots, j_{\pi(\ell)}$ to processor $P_1$, and the jobs
$ j_{\pi(n)}, j_{\pi(n-1)}, \hdots, j_{\pi(\ell+2)},j_{\pi(\ell+1)}$ to processor $P_2$, in this order. Note that then, one of the jobs $j_{\pi (\ell)}$, $j_{\pi(\ell +1)}$ ends at $d_{n-1}$, and the other at $-d_{n-1}+\sum_{i=1}^n j_i=d_n$, which is as desired.

%If \WSRK{2} has no solution, then \SCD has no solution either,
On the other hand, if \SCD has a solution, then \WSRK{2} has a solution as well,
because the previous construction is easily inverted. Visually, the list of jobs of $P_2$ is reversed and appended to the list of jobs of $P_1$. Job lengths correspond to vertex weights and deadlines correspond to splits (the last deadline where a job from $P_1$ finishes is merged with the last deadline where a job from $P_2$ finishes). Thus, \SCD is polynomial-time-reducible to \WSRK{2}.

\begin{lemma}\label{lem:redScdWsrk}
 \SCD $\le_p$ \WSRK{2}.
\end{lemma}

\newcommand{\dNMTS}{dNMTS\xspace}
\noindent
In the remainder of this section, we show that \dNMTS is polynomial-time-reducible to \SCD. The \dNMTS problem is equal to the \NMTS problem, except that all integers in $A\cup B\cup S$ are pairwise distinct. This variant has been shown to be strongly \classNP-hard by Hulett et al.~\cite{HulettWW08}. As the proof becomes somewhat simpler, we use \dNMTS instead of \NMTS for our reduction.

Let us first give a high level description of the main ideas of the reduction. For a \dNMTS instance $(A,B,S)$, the elements of $A\cup B$ will be encoded as jobs, and the elements of $S$ will be encoded as deadlines. A convenient way to represent an element $s\in S$ is by introducing segments which are delimited to the left and the right by double deadlines, and whose distance is equivalent to $s$. The elements of $A\cup B\cup S$ are blown up by well-chosen additive factors that preserve solutions and make sure that the length of each segment can only be met by the sum of exactly two job-lengths, one corresponding to an element of $A$ and the other to an element of $B$.

Our reduction will create an instance whose solution assigns, in each segment, one $x$-job (a job corresponding to an $A$-element) and one $y$-job (a job corresponding to a $B$-element) to the same processor, such that these two jobs are the only jobs executed on this processor in this segment, thus providing a solution to \dNMTS. \myWlog, the $x$-job is scheduled first. As we must not introduce any restriction which $x$-jobs can be assigned to which segments, we introduce a deadline for each length of an $x$-job; these are the real deadlines. We refer to the $x$- and $y$-jobs as green jobs. The job lengths were blown up such that in each segment, exactly one processor starts with a green $x$-job, and in each segment, exactly one processor ends by executing a green $y$-job. In each segment, the green jobs must not overlap; this is achieved by multiplying all deadlines created so far and the corresponding job lengths by a factor $2$, and introducing fake deadlines at odd positions one unit before the real deadlines. If an $x$-job and a $y$-job overlapped, there would be no job ending at the fake deadline preceding the real deadline at which the $x$-job ends, as all green jobs have even length and all real deadlines and double deadlines are even. Blue, red, and black jobs are created to meet all deadlines on the processor that is not currently executing green jobs. The blow-up of the elements of $A\cup B\cup S$ ensures that these jobs cannot equate the green jobs (except the black jobs whose lengths might equal the lengths of green $y$-jobs, but, \mywlog, one can assign them to the last part of each segment of the processor not executing a green job). That none of these jobs is executed between two green jobs within a segment is ensured as the sum of all green job lengths equals the sum of the lengths of the segments. This summarizes the reduction and gives the reasons for the different elements of the construction. Let us now turn to the formal reduction.

Let $(A,B,S)$ be an instance for \dNMTS. We suppose, \mywlog, that $\sum_{i=1}^m s_i = \sum_{x\in A \cup B} x$, otherwise $(A,B,S)$ is trivially a \textsc{No}-instance for \dNMTS. Let $A=\set{a_1, \hdots, a_m}$ and $B=\set{b_1, \hdots, b_m}$. We also assume, \mywlog, that $a_i < a_{i+1}$, $b_i < b_{i+1}$, $s_i < s_{i+1}$, for all $i \in \set{1, \hdots, m-1}$, that $a_m < b_m$, and that $s_m\le a_m+b_m$.

First, we construct an equivalent instance $(X,Y,Z)$ for \dNMTS. Each of $X:=\set{x_1, \hdots, \linebreak[1] x_n}$, $Y:=\set{y_1, \hdots, y_n}$, and $Z:=\set{z_1, \hdots, z_n}$ has $n:=m+1$ elements:
%%ML%%\begin{align*}
%%ML%% x_i &:= 2 \cdot (a_i + (b_m + 2)),\\
%%ML%% y_i &:= 2 \cdot (b_i + 3 \cdot (b_m + 2)),\\
%%ML%% z_i &:= 2 \cdot (s_i + 4 \cdot (b_m + 2)), &\text{ for } i\in \set{1, \hdots, n-1},\\
%%ML%% x_n &:= 2 \cdot (a_m + 1 + (b_m + 2)),\\
%%ML%% y_n &:= 2 \cdot (b_m + 1 + 3 \cdot (b_m + 2)), \text{ and}\\
%%ML%% z_n &:= 2 \cdot (a_m + b_m + 2 + 4 \cdot (b_m + 2)).
%%ML%%\end{align*}
\begin{align*}
 &\text{for } i\in \set{1, \hdots, n-1},&&\\
 &x_i := 2 \cdot (a_i + (b_m + 2)),&&\quad\quad\quad			x_n := 2 \cdot (a_m + 1 + (b_m + 2)),\\
 &y_i := 2 \cdot (b_i + 3 \cdot (b_m + 2)),&&\quad\quad\quad	y_n := 2 \cdot (b_m + 1 + 3 \cdot (b_m + 2)),\\
 &z_i := 2 \cdot (s_i + 4 \cdot (b_m + 2)),\text{ and}&&\quad\quad\quad	z_n := 2 \cdot (a_m + b_m + 2 + 4 \cdot (b_m + 2)).
\end{align*}
The elements of $X$, $Y$, and $Z$ have the following properties.

%\medskip
\begin{prop}\label{prop:even}
 Each element of $X\cup Y\cup Z$ is an even positive integer.
\end{prop}
\begin{prop}\label{prop:increasing}
 For every $i\in \set{1, \hdots, n-1}$, we have that $x_i < x_{i+1}$, that $y_i < y_{i+1}$, and that $z_i < z_{i+1}$.
\end{prop}
\begin{prop}\label{prop:ybiggerx}\label{prop:sumbig}\label{prop:ybig}\label{prop:rangexyz}
 For every $i\in \set{1, \hdots, n}$, we have
 \begin{align*}
  2\cdot b_m + 4  &\le x_i \le 4 \cdot b_m + 4,\\
  6\cdot b_m + 12 &\le y_i \le 8 \cdot b_m + 14, \text{ and}\\
  8\cdot b_m + 16 &\le z_i \le 12 \cdot b_m + 18.\\
 \end{align*}
\end{prop}

\noindent
In particular, Property \ref{prop:rangexyz} implies that $y_1 > x_n$, that $z_1 > y_n$, and that $2\cdot y_1 > z_n$. Properties \ref{prop:even}--\ref{prop:rangexyz} easily follow by construction of $X, Y$, and $Z$.
%
% \begin{prop}\label{prop:ybiggerx}
%  $y_1 > x_n$
% \end{prop}
% 
% \begin{prop}\label{prop:sumbig}
%  $z_1 > y_n$.
% \end{prop}
% 
% \begin{prop}\label{prop:ybig}
%  $2\cdot y_1 > z_n$.
% \end{prop}
%
\begin{prop}\label{prop:sumn}
 If $k$ and $\ell$ are integers such that $x_k+y_\ell=z_n$, then $k=\ell=n$.
\end{prop}

\noindent
Property \ref{prop:sumn} holds because $x_n$ and $y_n$ are the only elements of $X$ and $Y$, \longversion{respectively}\shortversion{resp.}, that are large enough to sum to $z_n$.

\begin{prop}\label{prop:oneXoneY}
 Let $p,q\in X\cup Y$, $p\le q$, and $z\in Z$. If $p+q=z$, then $p\in X$ and $q\in Y$.
\end{prop}

\noindent
By Property~\ref{prop:ybig}, the sum of any two $X$-elements is smaller and the sum of any two $Y$-elements is larger than any element of $Z$.

\medskip
\noindent
For our \SCD instance, we create the following deadlines:
\begin{itemize}
 \item \emph{real} deadlines: $r_{i,j} := x_i + \sum_{k=1}^j z_k$, for each $j \in \set{0, \hdots, n-1}$ and each $i \in \set{1, \hdots, n}$,
 \item \emph{fake} deadlines: $f_{i,j} := r_{i,j}-1$, for each $j \in \set{0, \hdots, n-1}$ and each $i \in \set{1, \hdots, n}$, and
 \item \emph{sum} deadlines: two deadlines $ds_{1,j} := ds_{2,j} := \sum_{k=1}^j z_k$, for each $j \in \set{1, \hdots, n}$.
\end{itemize}
The sum deadlines we just defined partition the interval $[0,ds_{1,n}]$ into $n$ segments $I_j:=[ds_{1,j-1},ds_{1,j}]$, $j=1,\ldots n$, where for convenience, we let $ds_{1,0}=0$.
We create jobs with the following lengths, where $x_0=0$ :
\begin{itemize}
 \item green x-jobs: $x_i$, for each $i \in \set{1, \hdots, n}$,
 \item green y-jobs: $y_i$, for each $i \in \set{1, \hdots, n}$,
 \item blue jobs: $n\cdot (n-1)$ times a job of length $1$,
 \item red fill jobs: $n-1$ times a job of length $x_{i} - 1 - x_{i-1}$, for each $i \in \set{1, \hdots, n}$,
 \item red overlap jobs: $x_{i} - x_{i-1}$, for each $i \in \set{1, \hdots, n}$,
 \item black fill jobs: $z_i - x_n$ for $i \in \set{1, \hdots, n-1}$, and
 \item a black overlap job: $z_n - x_n + 1$.
\end{itemize}

%\bigskip

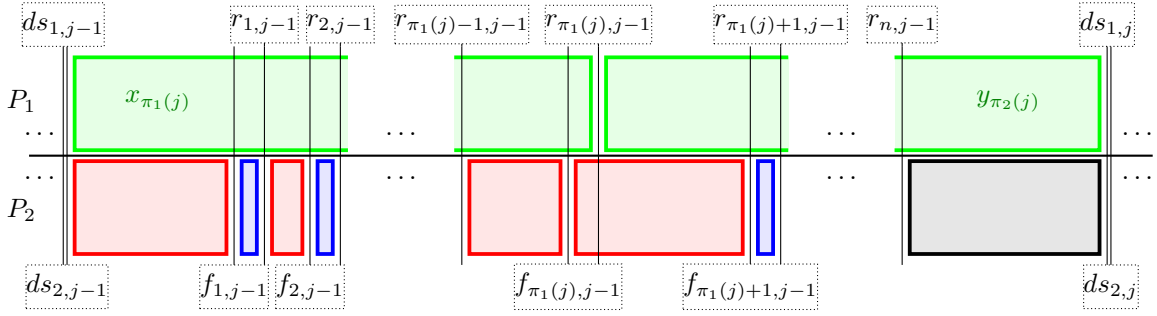
\begin{figure}[tb]
\begin{center}
\begin{tikzpicture}[xscale=2,yscale=1.45]

  \fill[green!10] (0.25,0.05) rectangle (2.05,0.9);
  \draw[green,line width=0.5mm] (2.05,0.9)-- (0.25,0.9)-- (0.25,0.05)-- (2.05,0.05);
  \fill[green!10] (2.75,0.05) rectangle (3.65,0.9);
  \draw[green,line width=0.5mm] (2.75,0.9)-- (3.65,0.9)-- (3.65,0.05)-- (2.75,0.05);
  \fill[green!10] (3.75,0.05) rectangle (4.95,0.9);
  \draw[green,line width=0.5mm] (4.95,0.9)-- (3.75,0.9)-- (3.75,0.05)-- (4.95,0.05);
  \fill[green!10] (5.65,0.05) rectangle (7,0.9);
  \draw[green,line width=0.5mm] (5.65,0.9)-- (7,0.9)-- (7,0.05)-- (5.65,0.05);
%  \filldraw[green!10, draw=green, line width=0.5mm]  (3.75,0.05) rectangle (7,0.9);
%  \filldraw[preaction={fill=red!10}] [red!10, draw=red, line width=0.5mm, pattern=crosshatch dots, pattern color=red!70]  (0.25,-0.05) rectangle (1.25,-0.9);
%  \filldraw[preaction={fill=blue!10}] [blue!10, draw=blue, line width=0.5mm, pattern=north east lines, pattern color=blue!70]  (1.35,-0.05) rectangle (1.45,-0.9);
%  \filldraw[preaction={fill=red!10}] [red!10, draw=red, line width=0.5mm, pattern=crosshatch dots, pattern color=red!70]  (1.55,-0.05) rectangle (1.75,-0.9);
%  \filldraw[preaction={fill=blue!10}] [blue!10, draw=blue, line width=0.5mm, pattern=north east lines, pattern color=blue!70]   (1.85,-0.05) rectangle (1.95,-0.9);
%  \filldraw[preaction={fill=red!10}] [red!10, draw=red, line width=0.5mm, pattern=crosshatch dots, pattern color=red!70]  (2.85,-0.05) rectangle (3.45,-0.9);
%  \filldraw[preaction={fill=red!10}] [red!10, draw=red, line width=0.5mm, pattern=crosshatch dots, pattern color=red!70]  (3.55,-0.05) rectangle (4.65,-0.9);
%  \filldraw[preaction={fill=blue!10}] [blue!10, draw=blue, line width=0.5mm, pattern=north east lines, pattern color=blue!70]  (4.75,-0.05) rectangle (4.85,-0.9);
%  \filldraw[preaction={fill=black!10}] [black!10, draw=black, line width=0.5mm, pattern=crosshatch, pattern color=black!70]  (5.75,-0.05) rectangle (7,-0.9);
  \filldraw[preaction={fill=red!10}] [red!10, draw=red, line width=0.5mm]  (0.25,-0.05) rectangle (1.25,-0.9);
  \filldraw[preaction={fill=blue!10}] [blue!10, draw=blue, line width=0.5mm]  (1.35,-0.05) rectangle (1.45,-0.9);
  \filldraw[preaction={fill=red!10}] [red!10, draw=red, line width=0.5mm]  (1.55,-0.05) rectangle (1.75,-0.9);
  \filldraw[preaction={fill=blue!10}] [blue!10, draw=blue, line width=0.5mm]   (1.85,-0.05) rectangle (1.95,-0.9);
  \filldraw[preaction={fill=red!10}] [red!10, draw=red, line width=0.5mm]  (2.85,-0.05) rectangle (3.45,-0.9);
  \filldraw[preaction={fill=red!10}] [red!10, draw=red, line width=0.5mm]  (3.55,-0.05) rectangle (4.65,-0.9);
  \filldraw[preaction={fill=blue!10}] [blue!10, draw=blue, line width=0.5mm]  (4.75,-0.05) rectangle (4.85,-0.9);
  \filldraw[preaction={fill=black!10}] [black!10, draw=black, line width=0.5mm]  (5.75,-0.05) rectangle (7,-0.9);

 \draw[thick] (-0.05,0)--(7.35,0);
 \draw (0.2,1)--(0.2,-1) (0.175,1)--(0.175,-1);
 \draw (1.5,1)--(1.5,-1) (1.3,1)--(1.3,-1)
            (2,1)--(2,-1) (1.8,1)--(1.8,-1)
            (2.8,1)--(2.8,-1)
            (3.7,1)--(3.7,-1) (3.5,1)--(3.5,-1)
            (4.9,1)--(4.9,-1) (4.7,1)--(4.7,-1)
            (5.7,1)--(5.7,-1)
            (7.075,1)--(7.075,-1) (7.05,1)--(7.05,-1);

\draw (1.5,1.2) node [draw,thin,densely dotted,fill=white,inner xsep=0pt] {$r_{1,j-1}$} -- (1.5,1);
\draw (2,1.2) node [draw,thin,densely dotted,fill=white,inner xsep=0pt] {$r_{2,j-1}$} -- (2,1);
\draw (2.8,1.2) node [draw,thin,densely dotted,fill=white,inner xsep=0pt] {$r_{\pi_1(j)-1,j-1}$} -- (2.8,1);
\draw (3.7,1.2) node [draw,thin,densely dotted,fill=white,inner xsep=0pt] {$r_{\pi_1(j),j-1}$} -- (3.7,1);
\draw (4.9,1.2) node [draw,thin,densely dotted,fill=white,inner xsep=0pt] {$r_{\pi_1(j)+1,j-1}$} -- (4.9,1);
\draw (5.7,1.2) node [draw,thin,densely dotted,fill=white,inner xsep=0pt] {$r_{n,j-1}$} -- (5.7,1);

\draw (1.3,-1.2) node [draw,thin,densely dotted,fill=white,inner xsep=0pt] {$f_{1,j-1}$} -- (1.3,-1);
\draw (1.8,-1.2) node [draw,thin,densely dotted,fill=white,inner xsep=0pt] {$f_{2,j-1}$} -- (1.8,-1);
\draw (3.5,-1.2) node [draw,thin,densely dotted,fill=white,inner xsep=0pt] {$f_{\pi_1(j),j-1}$} -- (3.5,-1);
\draw (4.7,-1.2) node [draw,thin,densely dotted,fill=white,inner xsep=0pt] {$f_{\pi_1(j)+1,j-1}$} -- (4.7,-1);

\draw (0.175,1.2) node [draw,thin,densely dotted,fill=white,inner xsep=0pt] {$ds_{1,j-1}$} -- (0.175,1);
\draw (0.2,-1.2) node [draw,thin,densely dotted,fill=white,inner xsep=0pt] {$ds_{2,j-1}$} -- (0.2,-1);
\draw (7.05,1.2) node [draw,thin,densely dotted,fill=white,inner xsep=0pt] {$ds_{1,j}$} -- (7.05,1);
\draw (7.075,-1.2) node [draw,thin,densely dotted,fill=white,inner xsep=0pt] {$ds_{2,j}$} -- (7.075,-1);

\draw[thick,green!50!black] (0.8, 0.5) node {$x_{\pi_1(j)}$};
\draw[thick,green!50!black] (6.4, 0.5) node {$y_{\pi_2(j)}$};

\draw[very thick] (7.25,0.2) node {$\mathbf{\hdots}$};
\draw[very thick] (7.25,-0.2) node {$\mathbf{\hdots}$};
\draw[very thick] (0.02,0.2) node {$\mathbf{\hdots}$};
\draw[very thick] (0.02,-0.2) node {$\mathbf{\hdots}$};
 
\draw[very thick] (2.4,0.2) node {$\mathbf{\hdots}$};
\draw[very thick] (2.4,-0.2) node {$\mathbf{\hdots}$};
\draw[very thick] (5.3,0.2) node {$\mathbf{\hdots}$};
\draw[very thick] (5.3,-0.2) node {$\mathbf{\hdots}$};

\draw (-0.1,0.5) node {$P_1$};
\draw (-0.1,-0.5) node {$P_2$};
\end{tikzpicture}
\caption{\label{fig:scd}How jobs are assigned to processors in the \SCD instance in segment $j < n$.}
%\longversion{(The patterns of the jobs are the following: Red jobs are dotted, blue jobs are lines, black jobs are crosshatch, and green jobs have no pattern.)}}
\end{center}
\end{figure}

\noindent
To illustrate these definitions, we start by showing that if we have
%an \SCD \textsc{No}-instance, then $(X,Y,Z)$ is a \textsc{No}-instance for \dNMTS as well.
a \textsc{Yes}-instance $(X,Y,Z)$ for \dNMTS, then we have an \SCD \textsc{Yes}-instance as well.
%Suppose otherwise, and 
Let $C_1, C_2, \hdots, C_n$ be $n$ couples such that $C_j = \set{x_{\pi_1(j)},y_{\pi_2(j)}}$ and $x_{\pi_1(j)}+y_{\pi_2(j)}=z_j$, $j\in \set{1, \hdots, n}$, for two permutations $\pi_1$ and $\pi_2$ of the set $\set{1, \hdots, n}$. We construct a solution for \SCD. Let us construct the schedules for $P_1$ and $P_2$. For each $j\in \set{1, \hdots, n-1}$,
\begin{itemize}
 \item assign the green $x$-job $x_{\pi_1(j)}$ to the interval $[ds_{1,j-1},r_{\pi_1(j),j-1}]$ of $P_1$,
 \item assign the green $y$-job $y_{\pi_2(j)}$ to the interval $[r_{\pi_1(j),j-1},ds_{1,j}]$ of $P_1$,
 \item assign a red fill job of length $x_1-1$ to the interval $[ds_{1,j-1},f_{1,j-1}]$ of $P_2$,
 \item for every $i\in \set{1,\hdots,n-1} \setminus \pi_1(j)$, assign a red fill job of length $x_{i+1}-1-x_{i}$ to the interval $[r_{i,j-1},f_{i+1,j-1}]$ of $P_2$,
 \item for every $i\in \set{1,\hdots,n} \setminus \pi_1(j)$, assign a blue job to the interval $[f_{i,j-1},r_{i,j-1}]$ of $P_2$,
 \item assign a red overlap job of length $x_{\pi_1(j)+1}-x_{\pi_1(j)}$ to the interval $[f_{\pi_1(j),j-1},\linebreak[1] f_{\pi_1(j)+1,j-1}]$ of $P_2$, and
 \item assign a black fill job of length $z_j-x_n$ to the interval $[r_{n,j-1},ds_{1,j}]$ of $P_2$.
\end{itemize}
It only remains to assign jobs to the last segment. The last segment of $P_1$ contains the green $x$-job $x_n$ and the green $y$-job $y_n$, in this order. The last segment of $P_2$ contains a red fill job of length $x_1-1$, a blue job, a red fill job of length $x_2-1-x_1$, a blue job, $\hdots$, a red fill job of length $x_n-1-x_{n-1}$, and the black overlap job, in this order.
See Fig. \ref{fig:scd} for an illustration.
%
%The construction of this solution contradicts the assumption that the \SCD instance is a \textsc{No}-instance. 

\smallskip

Now suppose the \SCD instance is a \textsc{Yes}-instance. We will show some structural properties of any valid assignment of jobs to the processors, which will help to extract a solution for our original \dNMTS instance.
We will show that in each segment $I_j$,
any valid solution for the \SCD instance has exactly one green x-job $x_k$ and exactly one green y-job $y_\ell$, and that 
$x_k$ and $y_\ell$ sum to $z_j$.

Consider a valid assignment of the jobs to the processors $P_1$ and $P_2$. As two jobs with the same length are interchangeable, when we encounter a job whose length belongs to more than one category (for example ``black fill'' and ``green y'') we may choose in this case, \mywlog, to which category the job belongs.

\begin{claimLNCS}\label{cl:blackfill}
A black fill job is assigned to each interval $[r_{n,j},ds_{1,j+1}]$ with $j\in \set{0,\hdots,n-2}$.
\end{claimLNCS}
\begin{proof}
 Let $j\in \set{0,\hdots,n-2}$. Two jobs must finish at the double deadline $ds_{1,j+1}, ds_{2,j+1}$. One of these must start at $r_{n,j}$ and thus has length $ds_{1,j+1}-r_{n,j} = \sum_{k=1}^{j+1} z_k - x_n - \sum_{k=1}^j z_k = z_{j+1} - x_n$. So this job is, \mywlog, a black fill job.
\end{proof}

\noindent
This uses up all black fill jobs.

\begin{claimLNCS}\label{cl:lastgreeny}
 The green y-job $y_n$ is assigned to the interval $[r_{n,n-1},ds_{1,n}]$.
\end{claimLNCS}

\begin{proof}
 As in the previous proof, one job must be assigned to this interval, whose length is $\sum_{k=1}^n z_k - x_n - \sum_{k=1}^{n-1} z_k = z_n-x_n$, which is $y_n$ by Property~\ref{prop:sumn}. Thus, the green y-job $y_n$ is assigned to the interval $[r_{n,n-1},ds_{1,n}]$.
\end{proof}

%\noindent
%The proofs of claims~\ref{cl:lastgreeny},~\ref{cl:blackoverlap} and~\ref{cl:redfill}
%are provided in appendix~\ref{appHardPath}.

%The proof of the next claim is given in appendix~\ref{appHardPaths}.

\begin{claimLNCS}\label{cl:blackoverlap}
 The black overlap job is assigned to the interval $[f_{n,n-1},ds_{1,n}]$.
\end{claimLNCS}

\begin{proof}
As $r_{n,n-1}$ is the only deadline between $f_{n,n-1}$ and $ds_{1,n}$, the processor that does not use this deadline needs to process a job finishing at $ds_{1,n}$ and starting before $r_{n,n-1}$. This is the black overlap job, since no other job is long enough. It is assigned to the interval  $[f_{n,n-1},ds_{1,n}]$ of length $ds_{1,n}-f_{n,n-1}=z_n-x_n+1$.
\end{proof}

\noindent
This uses up all black jobs. Now, the only jobs left whose length is between $6 b_m + 12$ and $8b_m+14$ are the green y-jobs $y_1, \hdots, y_{n-1}$.

\begin{claimLNCS}\label{cl:ytask}
 For each $\ell\in \set{1, \hdots, n-1}$, the green $y$-job $y_{\ell}$ is assigned to an interval $[r_{i,j-1}, ds_{1,j}]$ for some $i\in \set{1, \hdots, n-1}$ and $j\in \set{1, \hdots, n-1}$.
%  To each interval $[ds_{1,j}, ds_{1,j+1}]$, $j\in\set{0, \hdots, n-1}$, one green $y$-job is assigned, and this job ends at $ds_{1,j+1}$.
\end{claimLNCS}
\begin{proof}
Each job is assigned to an interval inside some segment, as the double deadlines prevent jobs to span more than one segment.
Suppose the green $y$-job $y_{\ell}$ is assigned to segment $p$. As $ds_{1,p} + y_{\ell} > ds_{1,p} + x_n$, by Properties \ref{prop:increasing} and \ref{prop:ybiggerx}, and the deadline following $r_{n,p} = ds_{1,p} + x_n$ is $ds_{1,p+1}$, it must be that the green $y$-job $y_{\ell}$ finishes at $ds_{1,p+1}$. Moreover, $ds_{1,p+1}-y_\ell$ is equal to a real deadline as $ds_{1,p+1}-y_\ell$ is even.
\end{proof}

\noindent
Each of the $2n$ jobs that have been assigned so far finish at a double deadline $ds_{1,j}, ds_{2,j}$. Thus, no other jobs may end at a double deadline.
%The proof of the next claim is provided in Appendix~\ref{appHardPaths}.

\begin{claimLNCS}\label{cl:redfill}
 A red fill job of length $x_{1}-1$ is assigned to each interval $[ds_{1,j}, f_{1,j}]$ with $0 \le j \le n-1$.
\end{claimLNCS}

\begin{proof}
Since both processors finish a job at deadline $ds_{1,j}$ (\longversion{respectively}\shortversion{resp.}, are initialized at time $ds_{1,0}=0$) and one of them finishes a job at the following deadline, which is $f_{1,j}$, we need to assign a job of length $f_{1,j}-ds_{1,j}=x_1-1$ to the interval  $[ds_{1,j}, f_{1,j}]$. \myWlog, this is one of the red fill jobs of length $x_1-1$.
\end{proof}

\noindent
This uses up all red fill jobs of length $x_1-1$.

\begin{claimLNCS}\label{cl:xtask}
For each $\ell\in \set{1, \hdots, n}$, the green $x$-job $x_{\ell}$ is assigned to an interval $[ds_{1,j}, r_{i,j}]$ for some $i\in \set{1, \hdots, n}$ and $j\in \set{0, \hdots, n-1}$.
%  To each interval $[ds_{1,j}, ds_{1,j+1}]$, $j\in\set{0, \hdots, n-1}$, one green $x$-job is assigned, and this job starts at $ds_{1,j}$.
\end{claimLNCS}
\begin{proof}
Suppose the green $x$-job $x_{\ell}$ is assigned to segment $p$. Notice that $x_{\ell} > r_{n,p}-f_{1,p}$. Indeed $r_{n,p}-f_{1,p}=x_n-x_1+1$ and, by construction, $x_n-x_1+1\le 2b_m$, whereas $x_{\ell}\ge 2b_m+4$. Moreover, $r_{n,p}$ is the latest deadline in $p$. So the green $x$-job $x_{\ell}$ starts at $ds_{1,p}$. Notice that $ds_{1,p}+x_\ell < ds_{1,p+1}$ and that $ds_{1,p}+x_\ell$ corresponds to a real deadline as $ds_{1,p}+x_\ell$ is even, but all fake deadlines are odd.
\end{proof}

\noindent
By Claims \ref{cl:lastgreeny}, \ref{cl:ytask}, and \ref{cl:xtask}, and since we have the same amount of segments as green $x$-jobs, \longversion{respectively}\shortversion{resp.} green $y$-jobs, we obtain that
each segment $I_j$, $1 \le j \le n$,, contains exactly one green $x$-job and exactly one green $y$-job.

\begin{claimLNCS}
For $j\in\set{1, \hdots, n}$,
the green $x$-job and the green $y$-job in the segment $I_j$ do not overlap.
\end{claimLNCS}
\begin{proof}
Suppose otherwise, that is, suppose there is a $j\in\set{1, \hdots, n}$ such that $I_j$ contains a green $x$-job, say $x_\ell$, and a green $y$-job, say $y_k$, that overlap (i.e. the intervals they are assigned to overlap). Since $x_\ell$ ends at a real deadline by Claim \ref{cl:xtask} and $y_k$ starts at a real deadline by Claim \ref{cl:ytask}, no job ends at the fake deadline situated at $ds_{1,j-1}+x_{\ell}-1$, which contradicts the validity of the \SCD solution.
% Because of the choice of the real and fake deadlines in the interval $I$, the size of the overlap has to be odd, say it has size $c$. On the other hand, because of Claims~\ref{cl:xtask} and~\ref{cl:ytask}, we have $ds_{1,j+1}-ds_{1,j}=x_\ell+y_k-c$. This is impossible, since all numbers but $c$ in this equation are even.
\end{proof}

\noindent
The last claim implies that in each segment $I_j$, $1\le j\le n$, there is a green $x$-job $x_{\ell_j}$ and a green $y$-job $y_{k_j}$ which together have the same size as the interval. Hence the couples $C_{j} = \set{a_{\ell_j},b_{k_j}}, 1\le j \le n$, form \longversion{the desired} \shortversion{a} solution of \dNMTS.
Thus, we have the following lemma.

\begin{lemma}\label{lem:redNmtsScd}
 \dNMTS $\le_p$ \SCD.
\end{lemma}

\noindent
We have assembled enough information to prove our main theorem.

\begin{theorem}
 \WSRK{2} is strongly \classNP-complete.
\end{theorem}

\begin{proof}
The theorem follows from the strong \classNP-hardness of \dNMTS, Lemmas \ref{lem:redScdWsrk} and \ref{lem:redNmtsScd}, and the membership of \WSRK{2} in \classNP, which is easily verified as the certificate is a path and an assignment of the splits to its edges, all of which can be encoded in polynomial space.
\end{proof}

%\shortversion{\noindent Due to space constraints, the proof of the following corollary is given in the appendix.}

\begin{restatable}{corollary}{RCorCaterpillars}\label{CorCaterpillars}
%ML%\begin{corollary}[$\star$] \label{CorCaterpillars}
\catFull\ is \classNP-complete.
%ML%\end{corollary}
\end{restatable}

\longversion{

\begin{proof}
It is clear that this problem, abbreviated \catShort, is in \classNP. To show that it is hard for \classNP, we reduce from \WSRK{2}.
Let $\iP'=(\omega'_1, \hdots, \omega'_{n-2}, s'_1, \hdots, s'_{n-3})$ be an instance of \WSRK{2}, where
$\omega'_i$, $1\leq i \leq n-2$, are the vertex weights and $s'_j$, $1\leq j \leq n-3$, are the splits.
We assume that all vertex weights and splits are upper bounded by a polynomial in $n$;
as \WSRK{2} is strongly \classNP-hard, it is still \classNP-hard with this restriction.
Define $\Omega:= 1+2\cdot \max \set{\omega'_i: 1\leq i \leq n-2}$. To simplify the argument, consider an auxiliary instance $\iP=(\omega_1, \hdots, \omega_n, s_1, \hdots, s_{n-1})$ of \WSRK{2} obtained from $\iP'$ by:
\begin{itemize}
\item augmenting the values of $s'_j$, $1\leq j \leq n-3$, by $\Omega$,
\item adding $\omega_{n-1}=\omega_{n}=\Omega$ to the multiset of weights,
\item adding $s_{n-2}=s_{n-1}=\Omega$ to the multiset of splits,
\item and finally, multiplying each value in $\iP$ by $n\Omega$ (so, for $1\leq i \leq n-3$, $\omega_i = \omega'_in\Omega$, and $s_i=(s'_i+\Omega)\Omega$).
\end{itemize}
It is not difficult to see that $\iP$ and $\iP'$ are equivalent.

Now let us create an instance $\iC$ of \catShort\ in the following way.
\begin{itemize}
\item replace each weight $\omega_i$, $1\leq i \leq n$, by $\omega_i$ copies of weight $1$,
\item for each $\omega_i$, $1\leq i \leq n$, add {\em auxiliary splits} $s_{f,i}=f$, $1\leq f \leq \omega_i - 1$,
\item keep the {\em original splits} $(s_1,\hdots,s_{n-1})$.
\end{itemize}
Notice that there are $\sum_{i=1}^{n}\omega_i$ vertices and $(\sum_{i=1}^{n}\omega_i) -1$ splits (i.e. edges) in total.

As one easily checks, if $\iP$ has a solution then $\iC$ has a solution.  Now suppose $\iC$ has a solution $C$.
Then, as $C$ is an instance for \catShort, it follows that $C$ is a caterpillar of maximum degree $3$ (with unbounded hair-length). Call $B$ the backbone of $C$. Let $B'\subseteq B$ be maximal such that its endvertices have degree $3$.

By construction $s_i>1$, $1\leq i \leq n-1$, and only the splits $s_{1,i}$, $1\leq i \leq n$, have value $1$. 
There are exactly $n$ such splits, and so, $C$ must have exactly $n$ leaves. %Note that this implies that $|E(B')|\geq n-3$.
 
Since there is no split of value $n \Omega^2 + 1$, each hair of $C$ has length at most $n \Omega^2$.
So, as $s_i>  n \Omega^2$ for $i=1,\ldots n-3$, we obtain that  the splits $s_1, ..., s_{n-3}$
are assigned to edges $b_1,\ldots, b_{n-3}$ in $E(B')$. Observe that the edges $b_1,\ldots, b_{n-3}$
induce a connected graph (i.e. a path), as all other splits are smaller than the minimum of the $s_i$, $i=1,\ldots n-3$. 

Let $P$ be the path formed by the edges $b_1,\ldots, b_{n-3}$ and let $u$ and $v$ be the two
endpoints of $P$. There is at most one vertex $y$ such that
if we look at the values of the splits of the edges from $u$ to $y$ (resp. from $v$ to $y$),
then they are strictly increasing. In addition, if two edges of $P$ share a vertex $x$, $x \neq y$,
then there must be a hair attached to $x$, because the splits associated to these two edges differ by more than $1$.
Furthermore, there are hairs $H_1$ and $H_2$ of length
$n \Omega^2$ attached to the first and to the last vertex on the backbone, as no two auxiliary splits
are large enough to add up to one of the original splits  $s_i$, $i=1,\ldots n-3$.
From the fact that $C$ has exactly $n$ leaves, it follows that the remaining hair has to be attached to $y$.
As a consequence, $E(B')=\{b_1,\ldots, b_{n-3}\}$.

%Now, if two edges $b_i$ share a vertex $v$, then there must be hair attached to $v$, because the splits 
%$s_i$, $i=1,\ldots n-3$ differ by more than $1$. Furthermore, there are hairs $H_1$ and $H_2$ of length
%$n \Omega^2$ attached to the first and to the last vertex on the backbone, as no two auxiliary splits
%are large enough to add up to one of the original splits  $s_i$, $i=1,\ldots n-3$.
%This implies that $E(B')=\{b_1,\ldots, b_{n-3}\}$, by the choice of $B'$, and the fact that $C$
%has exactly $n$ leaves.

Let $B''$ be equal to $B'$ augmented with the two edges to which the splits of value $n \Omega^2$ are assigned. 
All edges outside $B''$ (that is, edges from hairs) belong to auxiliary splits. This means that the edges adjacent to $B''$ correspond to auxiliary splits $s_{\omega_i-1,i}$.

In order to find a solution for $\iP$, it thus suffices to take $B''$ and replace all hairs with the corresponding weight on their starting vertex on $B''$.
\end{proof}
}

\section{Algorithm for \WSRK{2} with few distinct vertex weights}
\label{sec:algoPaths} % with a bounded number of

Let \longversion{$k = |\set{\omega(v)~:~v\in V}|$}\shortversion{$k = |\set{\omega(v):v\in V}|$} denote the number of distinct vertex weights in an instance
$(V,\omega, \mathcal{S})$ for \WSRK{2}.
\longversion{In this section, we}\shortversion{We} exhibit a dynamic programming algorithm for \WSRK{2} that works in polynomial time when $k$ is a constant.
Moreover, standard backtracking can be used to actually
construct a solution, if one exists.

\newcommand{\pivot}{v_{\text{pivot}}}

Suppose $|V|=n$ and the multiset of splits, $\mathcal{S}$, contains the splits
$s_1 \le s_2 \hdots \le s_{n-1}$. Let $w_1 < w_2 \hdots < w_k$ denote the distinct vertex weights
and $m_1, m_2, \hdots, m_k$ denote their respective multiplicities, i.e.
$m_i = |\set{v\in V~:~\omega(v)=w_i}|$ for all $i\in \set{1, 2, \hdots, k}$.
%Moreover, let $\pivot$ be a special vertex with weight $w_{k+1} = s_{|\mathcal{S}|} + 1$.

Our dynamic programming algorithm computes the entries of a boolean table $A$.
The table $A$ has an entry
$A[p, \linebreak[1] W_L,\linebreak[1]  W_R, \linebreak[1] v_1, \linebreak[1] v_2, \linebreak[1] \hdots, \linebreak[1] v_{k}]$
for each integer $p$ with $1 \leq p \leq n-1$, 
%split $s_p \in \mathcal{S}$,
each two integers $W_L,W_R \in \mathcal{S}$,
and each $v_i \in \set{0,1,\hdots, m_i}$, where $i\in \set{1,2, \hdots, k}$.
The entry
$A[p, \linebreak[1] W_L, \linebreak[1] W_R, \linebreak[1] v_1, \linebreak[1] v_2, \linebreak[1] \hdots, \linebreak[1] v_{k}]$
is set to \texttt{true} iff there is an assignment of the splits $s_1, s_2, \hdots, s_p$ to the $\ell$
leftmost edges and the $r$ rightmost edges of the path $P_n$ on $n$ vertices, such that
% the following properties are fulfilled :
\begin{itemize}
\item $p=\ell + r$;
\item $v_1$ weights $w_1$, $v_2$ weights $w_2$, \dots, and $v_k$ weights $w_k$ are assigned
to the $\ell$ leftmost and the $r$ rightmost vertices of $P_n$
such that each split assigned to the left (respectively to the right) part of the path corresponds to the sum of the vertex
weights assigned to vertices to the left (respectively to the right) of this split; and
\item$W_L$
is equal to the value of the $\ell^{\text{th}}$ split from the left and $W_R$ is equal to the
$r^{\text{th}}$ split from the right.
\end{itemize}
Intuitively, our algorithm assigns splits and weights by starting from both endpoints of the path
and trying to join these two sub-solutions.
%%%%%%%%%%%%%%%%%%%%%%%%%%%%%%%%%%%%%%%%%%%%%%%%%%%%%%%%%%%%%%%%%%%%%%%%%%%%%%%%%%%%%%%%%%%%%%%%%%%%%%%%%%%%%%%%%%%%%%%%%%%%%%%%%%%%%%%%%%%%%%%%
%%%%%%%%%%%%%%%%%%%%%%%%%%%%%%%%%%%%%%%%%%%%%%%%%%%%%%%%%%%%%%%%%%%%%%%%%%%%%%%%%%%%%%%%%%%%%%%%%%%%%%%%%%%%%%%%%%%%%%%%%%%%%%%%%%%%%%%%%%%%%%%%
%%%%%%%%%%%%%%%%%%%%%%%%%%%%%%%%%%%%%%%%%%%%%%%%%%%%%%%%%%%%%%%%%%%%%%%%%%%%%%%%%%%%%%%%%%%%%%%%%%%%%%%%%%%%%%%%%%%%%%%%%%%%%%%%%%%%%%%%%%%%%%%%
%\marginpar{can probably be extended to path with bounded number of leaves (and bounded distint weights) -> journal version}
%%%%%%%%%%%%%%%%%%%%%%%%%%%%%%%%%%%%%%%%%%%%%%%%%%%%%%%%%%%%%%%%%%%%%%%%%%%%%%%%%%%%%%%%%%%%%%%%%%%%%%%%%%%%%%%%%%%%%%%%%%%%%%%%%%%%%%%%%%%%%%%%
%%%%%%%%%%%%%%%%%%%%%%%%%%%%%%%%%%%%%%%%%%%%%%%%%%%%%%%%%%%%%%%%%%%%%%%%%%%%%%%%%%%%%%%%%%%%%%%%%%%%%%%%%%%%%%%%%%%%%%%%%%%%%%%%%%%%%%%%%%%%%%%%
%%%%%%%%%%%%%%%%%%%%%%%%%%%%%%%%%%%%%%%%%%%%%%%%%%%%%%%%%%%%%%%%%%%%%%%%%%%%%%%%%%%%%%%%%%%%%%%%%%%%%%%%%%%%%%%%%%%%%%%%%%%%%%%%%%%%%%%%%%%%%%%%

For the base case, set $A[0,W_L,W_R,v_1, v_2, \hdots, v_k]$ to \texttt{true}
if $W_L=W_R=v_1=v_2=\hdots =v_k=0$ and to \texttt{false} otherwise. We compute the remaining
entries of $A$ by increasing values of $p$ using the following recurrence.
\begin{align*}
A[p,W_L,W_R,v_1, v_2, \hdots, v_k] &= \bigvee_{i=1}^k
\begin{cases}
A[p-1,W_L-w_i,W_R,v_1,v_2, \hdots, v_{i-1},\\\quad\quad v_i-1,v_{i+1},v_{i+2}, \hdots, v_k]\\
\vee A[p-1,W_L,W_R-w_i,v_1,v_2, \hdots, v_{i-1},\\\quad\quad v_i-1,v_{i+1},v_{i+2}, \hdots, v_k]
\end{cases}
\end{align*}
In the previous recurrence, the formulas that refer to table entries that are undefined have the value \texttt{false}.

The final result of the algorithm is computed by evaluating the expression
\begin{equation*}
\bigvee_{\substack{W_L,W_R \in \mathcal{S}\\ i\in \set{1,2,\hdots , k}\\ (W_L \leq w_i+W_R)~\wedge~(W_R \leq w_i +W_L)}} \hspace{-1.4cm}A[|\mathcal{S}|, W_L,W_R,m_1,m_2,\hdots , m_{i-1},m_i-1,m_{i+1},m_{i+2},\hdots, m_k].
\end{equation*}
\shortversion{The correctness proof and the running-time analysis of our dynamic programming algorithm (see Appendix~\ref{ProofAlgoPaths}) establish the following theorem.}

\begin{restatable}{theorem}{RthAlgoPaths}\label{thAlgoPaths}
%ML%\begin{theorem}[$\star$]\label{thAlgoPaths}
 \WSRK{2} can be solved in time $O(n^{k+3}\cdot k)$, where $k$ is the number of distinct vertex
 weights of any input instance $(V,\omega, \mathcal{S})$ and \longversion{$n$ is the number of vertices.}\shortversion{$n=|V|$.}
%ML%\end{theorem}
\end{restatable}

\longversion{

\begin{proof}
The correctness of the base case is clear.
For the correctness of the recurrence, let
 $\pivot$ be the vertex on $P_n$ where the two sub-solutions corresponding to the left and to
 the right part of $P_n$ meet. First note that
the values of the splits increase from left to right until we encounter vertex $\pivot$, from
which point they decrease. Filling up the path from both ends, this means that reading the splits
from $s_1$ to $s_{n-1}$, we can assign them to the path, each time only deciding whether we assign
it to the left part or to the right part of the path (in the \SCD model, this would be equivalent
to deciding whether to meet the next deadline on the processor $P_1$ or on the processor $P_2$).
The first (respectively second) case of the recurrence corresponds to assigning the next split to
the left (respectively right) part of the path by inserting a vertex of weight
$w_i$, $i\in \set{1,2,\hdots , k}$. The correctness of the final evaluation follows because it
inserts the one missing vertex weight that has not been used between the left and the right part
of the path.

The table has $|\mathcal{S}|^3 \cdot \Pi_{i=1}^k (m_i+1) \le n^{k+3}$ entries, each entry can be
computed in time $O(k)$, and the final evaluation takes time $O(n\cdot k)$.
\end{proof}

}

\section{\SRK{3} is \classNP-complete}
\label{sec:hardnessDeg3}

In this section we show that \textsc{Splits Reconstruction} with unit weights
is \classNP-complete for trees with maximum degree $3$.
Our polynomial-time reduction is done from the strongly \classNP-complete \NMTS problem
recalled in Section~\ref{sec:prelim}. This problem remains \classNP-complete even
if each integer of the \NMTS instance is at most $p(m)$, where $p$ is a polynomial
and $m$ is the length of the description of the instance.
Let us just mention that the next theorem does not immediately follow from Corollary~\ref{CorCaterpillars}.

\begin{theorem}\label{th-SRK3}
 \SRK{3} is \classNP-complete.
\end{theorem}

\begin{proof}
Let
$\tilde{A}=\{\tilde{a}_1, \tilde{a}_2, \dots, \tilde{a}_{m}\}$,
$\tilde{B}=\{\tilde{b}_1, \tilde{b}_2, \dots, \tilde{b}_{m}\}$ and
$\tilde{S}=\{\tilde{s}_1, \tilde{s}_2, \dots, \tilde{s}_{m}\}$ be an instance of \NMTS. % with $|\tilde{A}|=|\tilde{B}|=|\tilde{S}|=m'$.
Let $C = \max \{x ~:~ x \in \tilde{A} \cup \tilde{B}\}$.
\myWlog,
%, we assume that the $\tilde{s}_i$'s are ordered increasingly and
we construct the following equivalent \NMTS instance:
\begin{align*}
	a_{i} &:= \tilde{a}_i +2+ 3C, \quad 1\leq i \leq m,\\
	b_{i} &:= \tilde{b}_i +3+ 5C, \quad 1\leq i \leq m, \text{ and}\\
	s_{i} &:= \tilde{s}_i +5+ 8C, \quad 1\leq i \leq m.
\end{align*}
Let $A = \bigcup_{1\leq i \leq m} \{ a_i\}$, $B = \bigcup_{1\leq i \leq m} \{ b_i\}$, and
$S = \bigcup_{1\leq i \leq m} \{ s_i\}$.
Clearly, the instance $(\tilde{A}, \tilde{B}, \tilde{S})$ has a solution \longversion{if and only if}\shortversion{iff}
the instance $(A,B,S)$ has a solution.

%\smallskip

Now we describe an instance $(V,\mathcal{S})$ of \SRK{3}, which is a \textsc{Yes}-instance \longversion{if and only if}\shortversion{iff}
the previous instance $(A,B,S)$ of \NMTS is a \textsc{Yes}-instance (see also Figure~\ref{fig-deg3}).

%\smallskip
%
%Let $n = 4m-2+2\sum_{i=1}^{m}a_i +2\sum_{i=1}^{m}b_i)$ be the number of vertices in the set $V$.
%We recall that these vertices have a unit weight.
%The multiset $\mathcal{S}$ of splits is defined as follows:
%\begin{itemize}
%\item To each value $s_i$, $1\leq i \leq m$, the value $1+s_i$ is added twice to $\mathcal{S}$
%and we refer to as \textit{red} splits;
%\item To each value $s_i$, $1< i < m$, the value $(i-1)+\sum_{j=1}^{i}(1+s_j)$ is added twice to $\mathcal{S}$;
%and we refer to as \textit{black} splits;
%
%\item The value $(m-1)+\sum_{j=1}^{m}(1+s_j)$ is added once to $\mathcal{S}$;
%and we refer to as the \textit{grey} split;
%
%\item To each value $a_i$, $1\leq i \leq m$, the values $\{1, 2, \dots, a_i\}$ are added twice to $\mathcal{S}$;
%and we refer to as \textit{green} splits.
%\item To each value $b_i$, $1\leq i \leq m$, the values $\{1, 2, \dots, b_i\}$ are added twice to $\mathcal{S}$;
%and we refer to as \textit{blue} splits.
%\end{itemize}
%Finally each values $x$ of $\mathcal{S}$ is replaced by $\min(x,n-x)$. We note that $\mathcal{S}$
%contains $n-1$ splits.

Let $n = 2m-2+\sum_{i=1}^{m}a_i +\sum_{i=1}^{m}b_i$ be the number of vertices in \longversion{the set }$V$;
we recall that \longversion{these vertices}\shortversion{they} have unit weight.
The multiset $\mathcal{S}$ of splits is defined as follows.
\begin{itemize}
\item For each value $s_i$, $1\leq i \leq m$, the value $1+s_i$ is added to $\mathcal{S}$
and we refer to these splits as \textit{red} splits.
\item For each value $s_i$, $2 \le i \le m-2$, the value $(i-1)+\sum_{j=1}^{i}(1+s_j)$ is added to $\mathcal{S}$
and we refer to these splits as \textit{black} splits.
\item For each value $a_i$, $1\leq i \leq m$, the values $\{1, 2, \dots, a_i\}$ are added to $\mathcal{S}$
and we refer to these splits as \textit{green} splits.
\item For each value $b_i$, $1\leq i \leq m$, the values $\{1, 2, \dots, b_i\}$ are added to $\mathcal{S}$
and we refer to these splits as \textit{blue} splits.
\end{itemize}
Finally each value $x$ of $\mathcal{S}$ is replaced by $\min(x,n-x)$.
%ML% When we refer to a split of value $x$, we mean a split of value $\min(x,n-x)$.
\longversion{As required, $\mathcal{S}$ contains $n-1$ splits.}%

\begin{figure}[tb]
\centering
\includegraphics[scale=0.45]{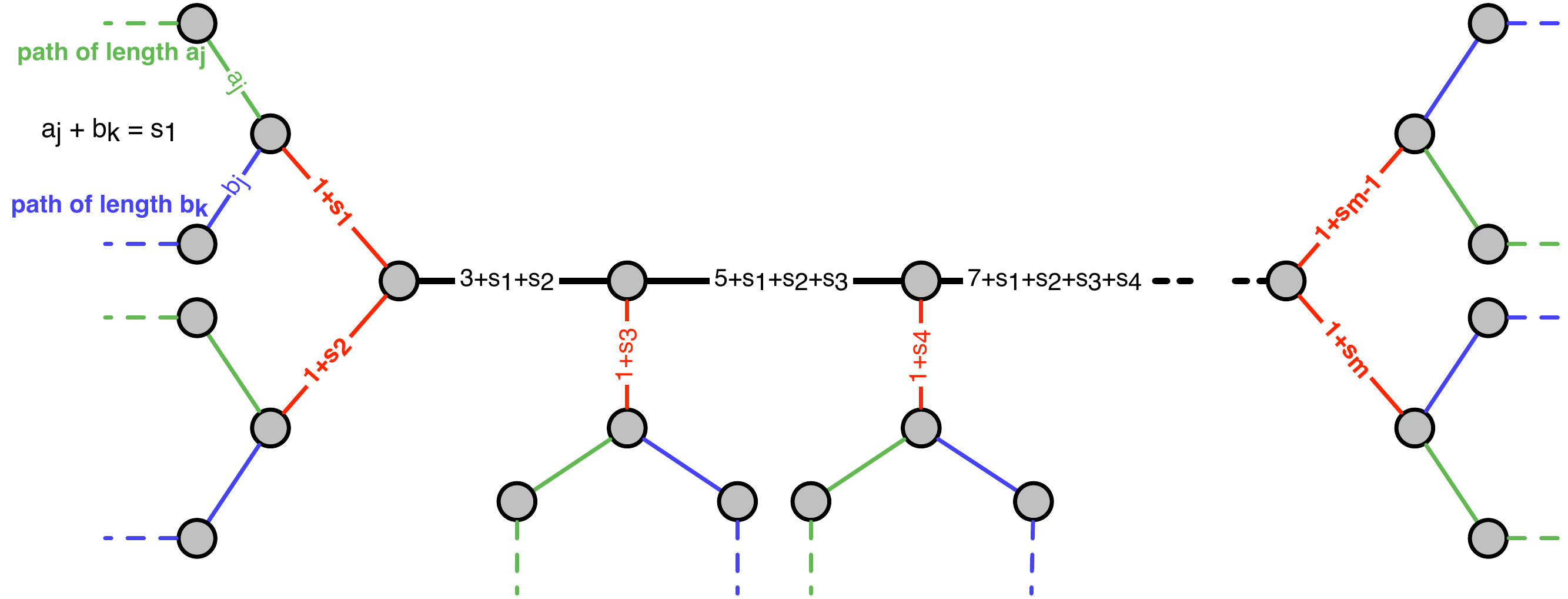}
\caption{\label{fig-deg3}A tree with maximum degree $3$ representing a solution to an \SRK{3} instance
constructed as described in the proof of Theorem~\ref{th-SRK3}.}
\end{figure}

%\medskip

%%ML%%We prove that $(A,B,S)$ is a \textsc{Yes}-instance for \NMTS if and only if $(V,\omega: V \rightarrow \set{1}, \mathcal{S})$
%%ML%%is a \textsc{Yes}-instance for \SRK{3}.

\begin{restatable}{lemma}{RepSR}\label{SR3NPhard}
%ML%\begin{lemma}[$\star$]\label{SR3NPhard}
$(A,B,S)$ is a \textsc{Yes}-instance for \NMTS if and only if $(V,\omega: V \rightarrow \set{1}, \mathcal{S})$
is a \textsc{Yes}-instance for \SRK{3}.
%ML%\end{lemma}
\end{restatable}

\longversion{

\begin{proof}%\shortversion{[of Lemma~\ref{SR3NPhard}]}
Throughout the proof, when we refer to a split of value $x$, we mean a split of value $\min(x,n-x)$.

\smallskip

``$\Rightarrow$'' Assume that $(A,B,S)$ is a \textsc{Yes}-instance for \NMTS. We will show that there is a solution to \SRK{3}.
A tree $T=(V,E)$ and a bijective function $b:E\rightarrow \mathcal{S}$ can be constructed as follows (see also Figure~\ref{fig-deg3}).
Construct a path $P$ with $m-3$ edges with the black splits such that the $(i-1)^{\text{th}}$ edge is associated
to the black split $(i-1)+\sum_{j=1}^{i}(1+s_j), i\in \set{2,3, \hdots, m-2}$. Add two edges incident to the first vertex of $P$,
that are associated to the red splits $1+s_1$ and $1+s_2$. Add two edges incident to the last vertex of $P$
that are associated to the red splits $1+s_{m-1}$ and $1+s_{m}$. To the $i^{\text{th}}$ vertex of $P$, $2 \le i \le m-3$,
add one incident edge associated to the red split $1+s_{i+1}$. Finally, for each $a_i \in A$ and
each $b_i \in B$, construct the paths with $a_i$ and $b_i$ vertices respectively. To each edge
of these $2m$ paths, we can associate a green or a blue split.
It remains to attach one green path and one blue path to each endpoint of an edge associated to a red split
(one endpoint is already involved in the path $P$ and of degree $3$). The way to attach these path
is given by the solution to the $(A,B,S)$ instance.

\smallskip

``$\Leftarrow$'' Assume that there exists a solution $(T, b)$ to \SRK{3}, where $T=(V,E)$ is a tree
of maximum degree $3$ and $b$ is a bijection from $E$ to $\mathcal{S}$. We show how a solution of the \NMTS instance $(A,B,S)$ can be derived from $(T, b)$.
Let us note that for any $i,j,k\in \set{1,2,\hdots,m}$, we have that $a_i + s_j > s_k$, that $a_i+a_j<s_k$, that $b_i+b_j>s_k$, and that $a_i+a_j>b_k$.
%The proofs of claims~\ref{cl:abpaths} and~\ref{cl:redsplit} are provided in appendix~\ref{appClaims}.
%\shortversion{The proof of Claims~\ref{cl:abpaths} and~\ref{cl:redsplit} are provided in Appendix~\ref{appClaims}.}
%
\begin{claimLNCS}\label{cl:abpaths}
 For every $i\in \set{1,2,\hdots,m}$, there is a path on $a_i$ edges, called the $a_i$-path, using the splits $1,2, \hdots, a_i$ (\mywlog, they are green) and there is a path on $b_i$ edges, called the $b_i$-path, using the splits $1,2, \hdots, b_i$ (\mywlog, they are blue). All these $a$-paths and $b$-paths are edge-disjoint.
\end{claimLNCS}

%\section{Proofs of Claims~\ref{cl:abpaths} and~\ref{cl:redsplit}}\label{appClaims}
%\section{Proof of Claim~\ref{cl:abpaths}}\label{appClaims}
%
%
\begin{proof}%\shortversion{[of Claim~\ref{cl:abpaths}]}
As the instance has $2m$ splits of value $1$, $T$ has $2m$ leaves. Each of these leaves is incident to a green or blue split of value $1$. As the instance also has $2m$ splits of each of the values $2,3, \hdots, 2+3C$, the leaves of $T$ are the starting points of $2m$ edge-disjoint paths $P_1, P_2, \hdots, P_{2m}$, each having $2+3C$ edges in $T$. Consider an $x\in A\cup B$ and the splits $2+3C+1, 2+3C+2, \hdots, x$. As $x<4+6C$, and as there is no split smaller than $2+3C$ other than those we have already used to form the paths $P_i, 1\le i\le 2m,$ the splits $2+3C+1, 2+3C+2, \hdots, x$ are assigned to an extension of a path $P_i$, which, together with $P_i$, forms a path $P_i'$ with $x$ edges using the splits $1,2, \hdots, x$. All these paths $P_i', 1\le i \le 2m,$ are edge disjoint and \mywlog, green splits are assigned to their edges if they have at most $2+4C$ edges and blue splits otherwise.
\end{proof}

\begin{claimLNCS}\label{cl:redsplit}
 For every $i\in \set{1,2,\hdots,m}$, the red split of value $1+s_i$ is assigned to an edge $e_i$ of $T$ whose vertex $u_i$ is the common extremity of an $a$-path and a $b$-path, where $u_i$ is in the subtree of $T-e_i$ that has $s_i+1$ vertices.
\end{claimLNCS}

\begin{proof}%\shortversion{[of Claim~\ref{cl:redsplit}]}
 As no split has value $s_i$, vertex $u_i$ is incident to another two edges besides $e_i$. We note that all splits, besides those of the $a$- and $b$-paths, have value at least $6+8C$. One such split plus the smallest $a_j$, $1\le j\le m$, would exceed $s_i$. So, $u_i$ is the end point of two $a/b$-paths. These cannot be two $a$-paths as $a_j+a_k < s_i$, for any $j,k\in \set{1,2,\hdots,m}$ and they cannot be two $b$-paths as $b_j+b_k>s_i$, for any $j,k\in \set{1,2,\hdots,m}$. Thus, $u_i$ is the common extremity of an $a$-path and a $b$-path.
\end{proof}

%\begin{proof}
% As no split has value $s_i$, vertex $u_i$ is incident to another two edges besides $e_i$. We note that all splits, besides those of the $a$- and $b$-paths, have value at least $6+8C$. One such split plus the smallest $a_j$, $1\le j\le m$, would exceed $s_i$. So, $u_i$ is the end point of two $a/b$-paths. These cannot be two $a$-paths as $a_j+a_k < s_i$, for any $j,k\in \set{1,2,\hdots,m}$ and they cannot be two $b$-paths as $b_j+b_k>s_i$, for any $j,k\in \set{1,2,\hdots,m}$. Thus, $u_i$ is the common extremity of an $a$-path and a $b$-path.
%\end{proof}

\noindent
Finally, a solution to the instance $(A,B,S)$ of \NMTS is formed by the couples $C_1,C_2, \hdots, \linebreak[1] C_m$, where each $C_i$ contains $a_{i_a}$ and $b_{i_b}$, where $i_a$ and $i_b$ are such that the edge $e_i$ of $T$, with $b(e_i) = 1+s_i$, is incident to the $a_{i_a}$-path and the $b_{i_b}$-path. This proves the \classNP-hardness of \SRK{3}.
\end{proof}

}

%The proof of Lemma~\ref{SR3NPhard} is provided in Appendix~\ref{appClaims}.
\noindent
As the certificate is a tree on $n$ vertices, the
membership in \classNP\ is obvious\longversion{ and Theorem~\ref{th-SRK3} is proved}.
\end{proof}

\section{Algorithm for \SR with few leaves}
\label{sec:algoLeaves} % with a bounded number of

In this section we design an algorithm for \SR parameterized by the number $k$ of splits that are equal to one, i.e.
$k = |\{s=1 : s\in \mathcal{S}\}|$.
As each such split is incident to a leaf in a reconstructed tree,
the algorithm reconstructs trees with $k$ leaves.

The algorithm starts with a star $T$ with center $r$ and $k$ leaves.
The vertex $r$ is also the root of $T$ and $r$ is the only vertex which is allowed to have non-unit weight
during the execution of the algorithm. We start by setting $\omega(r)=n-k$.
The splits that are equal to $1$ are assigned to the edges of the star. 

At any stage of the algorithm, $T$ is a tree with splits from $\mathcal{S}$ assigned to its edges, and the goal is to
replace the root $r$ of $T$ by a tree $T_r$ with unit-weight vertices (except for the new root, that can have a non-unit weight), 
using splits from $\mathcal{S}$ that have not been assigned yet; the leaves of $T_r$ are made adjacent to the former neighbors of $r$ in $T$.
%This tree replacing $r$ will be glued upon $T\setminus \{r\}$ such that the neighbors of $r$ in $T$ will be adjacent to the leaves of the tree that is glued upon $T\setminus \{r\}$.
If there exists such a replacement where the splits form a subset of $\mathcal{S}$, we say that $T$ has a \emph{valid extension}.
Each tree $T$ uniquely defines a partition $(A,C,U)$ of the splits $\mathcal{S}$ such that
\begin{itemize}
\item $A$ represents the multiset of \emph{available} splits that have not yet been assigned to $T$,
\item $C$ represents the multiset of \emph{current} splits assigned to edges incident to $r$, and
\item $U$ represents the multiset of \emph{used} splits assigned to edges of $T$ that are not incident to $r$.
\end{itemize}
Let $b$ denote the value of the smallest split in $C$.
Our tree $T$ will grow out of $r$ as follows.

\begin{itemize}
\item If $\omega(r)=1$, then return \textsc{True}.
 Indeed, $T$ uses all splits from $\mathcal{S}$ as $A$ is empty.
\item If $A$ contains a split whose value is at most $b$, then $T$ has no valid extension and the algorithm backtracks.
 Indeed, on a path between two leaves in a valid tree, there is no split with value at most $b$ between two splits with value at least $b$.
\item If $|\{s\in A : s=b+1\}| > |\{ s\in C : s=b \}|$, that is, $A$ contains more splits with value $b+1$ than $C$ contains splits with value $b$, then
$T$ has no valid extension and the algorithm backtracks.
The correctness of this case holds by the pigeonhole principle and the argument used in the previous case.
\item If $|\{s\in A : s=b+1\}| = |\{ s\in C : s=b \}|$, then all valid extensions of $T$ are also valid extensions of the tree obtained from $T$ by subdividing each edge with
split $b$ that is incident to $r$. That is, for each edge $r v$ with a split of value $b$, add a new vertex $z_v$, remove the edge $r v$, and add edges $r z_v$ and $z_v v$.
Decrement $\omega(r)$ by $|\{ s\in C : s=b \}|$.
The algorithm recursively solves the problem on this tree.
\item Otherwise (if $|\{s\in A : s=b+1\}| < |\{ s\in C : s=b \}|$), some split from $C$ with value $b$ receives a parent split with value more than $b+1$. Go over all choices
for selecting a subset $U$ of $N(r)$ of size at least $2$ containing a vertex $v$ such that $r v$ is associated with a split with value $b$.
If $A$ contains no split that equals $1+\sum_{u\in U} s(r u)$, where $s(e)$ denotes the split associated to the edge $e$ of $T$, then discard this choice.
Otherwise, create a new vertex $z_U$, remove the edges $\{r u : u\in U\}$ from $T$, add the edges $\{z_U u : u\in U\cup \{r\}\}$, and decrement $\omega(r)$ by $1$.
The algorithm resursively solves the resulting subproblems. If one such tree has a valid extension, $T$ has a valid extension.
\end{itemize}

\begin{theorem}
 \SR can be solved in time $O(8^{k \log k}\cdot n)$, where $k = |\{s=1 : s\in \mathcal{S}\}|$ and $n$ is the number of vertices.
\end{theorem}
\begin{proof}
 The arguments for correctness have been given in the description of the algorithm.
 For the running time analysis, we observe that $\omega(r)$ decreases in each recursive call, no recursive call increases $|C|$, and
 the time spent in each recursion step is linear. Let $T(c)$ denote an upper bound on the number of atomic instances solved
 for an instance with $|C|=c \le k$, where an instance is atomic if the algorithm makes no recursive call for solving the instance.
 In the only case making more than one recursive call, we have
 \begin{align*}
  T(c) \le \sum_{i=2}^{c} \binom{c}{i} T(c-i+1),
 \end{align*}
 as the set $U$ in the neighborhood of $N(r)$ is replaced by one vertex $z_U$. As $c\le k$ and $\binom{c}{i} \le k^i$, we have that
 \begin{align*}
 T(c) &\le (c-1)\cdot \max_{i=2..c} \left\{ c^i \cdot T(c-(i-1)) \right\} \\&\le \max_{i=2..c}  \left\{ k^{i+1} \cdot T(k-(i-1)) \right\} \\&\le \max_{i=2..c}  \left\{ k^{(i+1)\frac{k}{i-1}} \right\}\enspace.
 \end{align*}
 This maximum is attained for $i=2$, which proves the theorem.
\end{proof}

\newcommand{\ChWSR}{\textsc{ChWSR}\xspace}

\section{Freely choosable weights}
\label{sec:freeweights}

We remark that the following modification of \WSR makes any set of splits realizable in some tree.
Suppose the weight function $\omega$ is not given, but freely choosable, that is, we ask whether,
given a multiset $\mathcal{S}$ of integers, there exists a tree $T=(V,E)$ and a weight function $\weights$,
such that $\mathcal{S}$ is the multiset of splits of $T$. We call this problem \ChWSR.
%Let us denote by $\kappa$ the maximal multiplicity in $\mathcal{S}$.

\begin{restatable}{theorem}{RalwaysSol}\label{alwaysSol}
%ML%\begin{theorem}[$\star$]\label{alwaysSol}
\ChWSR always admits a solution.
%ML%\end{theorem}
\end{restatable}

\begin{proof}%\shortversion{[of Theorem~\ref{alwaysSol}]}
We show that the answer to \ChWSR is always yes: Decompose $\mathcal S$ into $\kappa$ chains $s_1^j<s_2^j<\ldots s_{m(j)}^j$, $j=1,\ldots ,\kappa$, where $\kappa$ is the maximal multiplicity in $\mathcal S$.
Let $T$ be obtained from the star $K_{1,\kappa}$ by subdividing $e_j$, the $j^\text{th}$ edge of $T$, $m(j)-1$ times (for $j=1,\ldots ,\kappa$), and root $T$ at the center $r$ of $K_{1,\kappa}$. Map $s_i^j$ to the edges of the subdivided $e_j$, $1\leq i \leq  m(j)$, keeping their order, so that the edge corresponding to $s_1^j$ is incident to a leaf of $T$. Finally, choose the weight $\omega (r)$ for the root to be equal to the maximum value in $\mathcal S$. For each leaf $v$ of $T$, set the weight $\omega (v)$ equal to the split assigned to the edge $\{v,u\}$, where $u$ is the parent of $v$. Any other vertex $v$ is given a weight equal to the difference of splits assigned to the edges incident to $v$. This choice of $T$ and $\omega$ clearly satisfies the requirements.
\end{proof}

%The proof of this theorem is provided in Appendix~\ref{appFreeWeights}.

\noindent \textbf{Remark.}
Due to the construction provided by the proof of Theorem~\ref{alwaysSol}, we
note that we are not only always able to construct a tree $T$ as required,
but the structure of this tree is also rather simple.
In particular, 
% Serge: the commented line below is not true if this number is one
%we can bound the maximum degree of $T$ by the maximum repetition of an element in $\mathcal S$. So,
the realization of the split sequence is a path if each split in $\mathcal S$ repeats at most twice.

\longversion{\smallskip

\noindent
Observe that if we consider \ChWSR with unit weights, we are back at the problem \SR. It is not difficult to see that in \SR, a given set of splits can be realized in the same way as explained in the proof of Theorem~\ref{alwaysSol} for \ChWSR, the only difference being that each time a non-unit weight $w$ is assigned to some vertex $v$ in \ChWSR, in \SR we have to add $w-1$ leaves of unit weight to $v$. Thus, if $\mathcal S$ contains a sufficient number of splits $1$, then $\mathcal S$ can always be realized by a tree. More precisely, setting the boundary values $s^j_0:=0$ for all $j$, and letting $\kappa$ denote the maximum multiplicity over all elements in $\mathcal{S}$ except $1$, we have that if $\kappa\ge 2$ and $\mathcal S$ contains at least 
\begin{equation*}
\kappa+\sum_{j=1}^{\kappa}\sum_{i=1}^{m(j)}(s^j_i-s^j_{i-1}-1) + 2\cdot \max_{1\le i \le \kappa} \set{ s_{m(i)}^i} -1 -\sum_{j=1}^{\kappa} s_{m(j)}^j 
\end{equation*}
times the split $1$, then $\mathcal S$ can be realized by a tree $T$: $\kappa$ of them are needed to be assigned to edges incident to leaves of the star, $\sum_{i=1}^{m(j)}(s^j_i-s^j_{i-1}-1)$ of them are added, with pending vertices, to vertices introduced by subdividing the edge $e_j$, and $2\cdot \max_{i=1}^{\kappa} \set{ s_{m(i)}^i }-1 -\sum_{j=1}^{\kappa} s_{m(j)}^j $ of them are added, with pending vertices, to the root. (Note that it does not matter if there are more splits $1$ than needed in our construction, since we may always add leaves to the center of $T$.) The previous condition is, of course, sufficient, but not necessary. Moreover, the tree $T$ that realizes $\mathcal S$ is a subdivided star to which some leaves have been added.  In particular, if each split in $\mathcal S$ repeats at most twice, then we can realize $\mathcal S$ in a caterpillar with hair-length one. We note that the conditions $\kappa=2$ and the lower bound on the number of splits with value $1$ are also necessary for caterpillars with hair length one.
}

\section{Conclusion}
\label{sec:concl}

In Section \ref{sec:algoPaths},  we \longversion{have shown}\shortversion{have shown} that
\WSRK{2} is in \classXP\ when parameterized by the number of distinct vertex weights.
It remains open whether this problem is fixed parameter tractable (a generalization of the problem is \classW{1}-hard \cite{FellowsGR10}).
%It is natural to ask whether the problem is in \classFPT\ or maybe hard for \classW{1}.
For practical purposes, it would further be important to identify other quantities that are small 
in practice (e.g. the diameter of the tree or topological indices),
%and investigate the multivariate complexity of the problems considered in this paper parameterized
and investigate the multivariate complexity of the considered problems parameterized
by combinations of these quantities.

There is a large contrast between the complexities of \WSR, where we are
given $n$ vertex weights, and \ChWSR, where we can freely choose the vertex weights, or,
alternatively, we can choose the vertex weights from an infinite multiset containing $n$ times
each element of $\mathbb{N}$.
%In \WSR, we are given $n$ vertex weights, $n$ being the number of vertices.
%This problem is \classNP-hard. In \ChWSR, we can freely choose the vertex weights, or,
%alternatively, we can chose the vertex weights out of an infinite multiset containing $n$ times
%each element of $\mathbb{N}$. This is trivial.} %
It would be interesting to know some
restrictions on the multiset of vertex weights such that the problem becomes\longversion{ polynomial time
solvable, or fixed-parameter} tractable with respect to interesting parameterizations\longversion{, when we can chose the weights from this multiset}. Ideally, these restrictions should be consistent with the applications in drug design and discovery.

\longversion{\bigskip}
\medskip
%\paragraph{Acknowledgment.} We thank Ming-Yang Kao for communicating this problem.

\noindent \textbf{Acknowledgment.} We thank Ming-Yang Kao for communicating this problem.

%\vspace{-0.49cm}

\bibliographystyle{plain}
\bibliography{splits}

\end{document}